\documentclass[12pt]{amsart}
\usepackage{amsmath}
\usepackage{amssymb}
\usepackage[latin1]{inputenc}
\usepackage[T1]{fontenc}
\usepackage{graphicx}
\usepackage{bbm,bm,dsfont} 
\newtheorem{definition}{Definition}
\newtheorem{lemma}{Lemma}
\newtheorem{remark}{Remark}
\newtheorem{theorem}{Theorem}

\newtheorem{proposition}{Proposition}
\newtheorem{example}{Example}

\newcommand{\vp}{\varphi}

\newcommand{\C}{\mathbb{C}}
\newcommand{\N}{\mathbb{N}}
\newcommand{\R}{\mathbb{R}}
\newcommand{\T}{\mathbb{T}}
\newcommand{\Z}{\mathbb{Z}}
\newcommand{\be}{\begin{equation}}
\newcommand{\eeq}{\end{equation}}
\newcommand{\bet}{\begin{equation*}}
\newcommand{\eeqt}{\end{equation*}}
\newcommand{\bea}{\begin{eqnarray}}
\newcommand{\eeqa}{\end{eqnarray}}
\newcommand{\beat}{\begin{eqnarray*}}
\newcommand{\eeqat}{\end{eqnarray*}}

\newcommand{\h}[1]{\mathcal{#1}}
\newcommand{\hil}{\mathcal{H}}

\newcommand{\hi}{\mathcal{H}}
\newcommand{\ki}{\mathcal{K}}

\newcommand{\lh}{\mathcal{L(H)}}
\newcommand{\lk}{\mathcal{L(K)}}
\renewcommand{\th}{\mathcal{T(H)}}
\newcommand{\tc}{\mathcal{T(H)}}

\newcommand{\sfq}{\mathsf{Q}}

\newcommand{\E}{\mathsf{E}}
\newcommand{\F}{\mathsf{F}}
\newcommand{\G}{\mathsf{G}}
\newcommand{\A}{\mathsf{A}}
\newcommand{\Ecan}{\E_\mathrm{can}}
\renewcommand{\P}{\mathsf{Q}} 

\newcommand{\tr}[1]{\mathrm{tr}\left[ {#1} \right]} 
\def\<{\langle} 
\def\>{\rangle} 
\newcommand{\M}{{\mathcal M}}
\newcommand{\I}{{\mathcal I}}

\newcommand{\kb}[2]{|#1\,\rangle\langle\,#2|} 

\begin{document}
\title{On measurements of the canonical phase observable}

\author{Juha-Pekka Pellonp\"a\"a}
\address{Turku Centre for Quantum Physics, Department of Physics and Astronomy, University of Turku, FI-20014 Turku, Finland}
\email{juhpello@utu.fi}

\author{Jussi Schultz}
\address{Dipartimento di Matematica, Politecnico di Milano, Piazza Leonardo da Vinci 32, I-20133 Milano, Italy}
\email{jussi.schultz@gmail.com}

\begin{abstract}
Measurements of single-mode  
phase observables are studied in the spirit of the quantum theory of measurement.
We determine the minimal measurement models of phase observables and consider methods of measuring such observables by using a double homodyne detector.
We show that, in principle, the canonical phase distribution of the signal state can be measured via double homodyne detection by first processing the state using a two-mode unitary channel.
\end{abstract}
\maketitle

\section{Introduction}

In quantum optics, the concept of phase for a single mode electromagnetic field  
has remained a somewhat controversial topic.
 Alternative descriptions for phase observables have been developed, and
hundreds of articles (see, e.g., the reference lists of \cite{Ly,PeBa9,Pel02,PIO,T48}) and several monographs  (e.g., \cite{MAWQAP,PIO})
have been written on the subject 
since Dirac's famous paper \cite{Di} published in 1927. 
A major reason for this variety of phase theories is that 
in trying to define the phase of a quantum oscillator, one can clearly 
see the restrictions of the conventional
approach which identifies observables with selfadjoint operators, or
equivalently, their spectral measures. In fact, it is well known that no
spectral measure satisfies all physically relevant conditions posed on
phase observables (see, e.g.,\ \cite{PSAQT, Pel02}),
and this problem has been
overcome with the introduction of the more general concept of
observables as (normalized) positive operator valued measures (POVMs).

A natural requirement for the description of a phase measurement is {\em covariance} with respect to phase shifts. In other words, the application of a phase shifter on the field prior to the measurement should only shift the phase distribution without changing its shape. 
Although there exist an infinite number of phase shift covariant POVMs, 
it is generally accepted that the canonical phase measurement for the single-mode radiation field is represented by the {London phase distribution} \cite{London}.
Hence, the canonical phase measurement is described by the canonical phase observable $\Ecan$, 
\begin{equation*}
\Ecan(X)=\sum_{m,n=0}^\infty\frac{1}{2\pi}\int_Xe^{i(m-n)\theta}d\theta \ \kb{m}{n} = \frac{1}{2\pi}\int_X  \kb{\theta}{\theta} d\theta 
\end{equation*}
where $X\subseteq [0,2\pi)$, $\{|m\>\}$ is the number basis,
and $|\theta\>=\sum_{m=0}^{\infty} e^{im\theta} |m\>$
is the (formal) {Susskind-Glogower phase state} \cite{SuGl}. 
We also recall that the canonical phase measurement arises as the limiting distribution of the Pegg-Barnett formalism \cite{PeBa9}.
In addition, $\Ecan$ has been independently derived by Helstrom \cite{Helstrom} and Holevo \cite{PSAQT} in the more general context of quantum estimation theory.

The canonical phase has a number of properties which makes it an optimal choice among other phase observables:
For instance, $\Ecan$ is pure, i.e.\ an extreme point of the convex set of all POVMs \cite{HePe}. Any other covariant phase observable $\F$ is connected to the canonical phase via a quantum channel $\Phi$ as $\F(X) = \Phi^*\left(\Ecan(X)\right)$ \cite{Ha2}.
The canonical phase 
 is also (essentially) the only covariant phase observable which generates
number shifts \cite{LaPe00}.
Furthermore, 
$\Ecan$ and the photon number $N$ are noncoexistent, probabilistically and
value complementary observables \cite{BuLaPeYl01}.
Finally, the canonical phase distribution of coherent states $|\alpha\>$, $\alpha\in\C$,
tends to a Dirac delta distribution in the classical limit $|\alpha|\to\infty$ 
 \cite{LaPe00}.  A list of further properties of $\Ecan$ can be found, e.g., in p.\ 51 of \cite{Pel02}.

The problem of finding a suitable realistic measurement model for the canonical phase observable is the last big open problem concerning the quantum description of the phase of an electromagnetic field. Of course, $\Ecan$ (or any observable of the field) can be `measured' indirectly or `sampled' by first measuring some informationally complete observable (e.g.\ the \ $Q$-function or the collection of rotated quadratures) and then constructing the canonical phase distribution.  However, these approaches may hardly be regarded as direct measurements in the spirit of the quantum theory of measurement \cite{BuLaMi}. 
Some suggestions for direct measurements of $\Ecan$ can be found from literature, see e.g.\ \cite{PeBa8}, but they have not led to experiments.
It seems that from the experimental point of view, the most easily accessible phase observables are those arising as the angle margins of certain (translation covariant) phase space
observables. Such observables are often referred to as (covariant) phase space phase observables, the most familiar example being the angle margin of the Husimi $Q$-function of the field \cite{Hu}. 
A natural measurement setup for these phase space measurements is then the double homodyne detection scheme (also known as eight-port homodyne detection) (see, e.g., \ \cite{MQSL}), which has also been demonstrated experimentally  \cite{WaCa}. Even though the phase space phase observables seem like a natural choice for measuring the phase of the electromagnetic field, they suffer from certain drawbacks when compared with the canonical phase. For instance, the phase space phase observables are not pure \cite{CaHePeTo}, and the canonical phase $\Ecan$ gives smaller (minimum) variance in large amplitude
coherent states than phase space phase observables  \cite{LaPe00}.

The purpose of this paper is to take a step in the direction of obtaining a realistic measurement scheme for the canonical phase. 
First, we determine those measurement models of covariant phase observables which are minimal (in a sense to be defined). In particular, we obtain as a special case the minimal measurement models of the canonical phase $\Ecan$. Since previous works indicate that the natural candidates for  phase measurements are  actually measurements of some phase space observables, which would then give the phase distribution after integration over the radial part, we then study joint measurements of arbitrary POVMs and the canonical phase.
After that we focus on the so-called phase shift covariant phase space observables, where we also present some physically relevant examples.
Finally, we turn to the double homodyne detection scheme and show that, in principle, the canonical phase distribution of the signal state can be measured with the setup by adding an extra two-mode unitary coupling in front of the apparatus.

\section{Preliminaries}

Throughout this paper we consider a single mode electromagnetic field as our physical system. The Hilbert space of the system is therefore $\hil\simeq L^2 (\R)$, and it is spanned by the orthonormal basis consisting of the number states $\vert n\rangle $, $n\in\N = \{0,1,2,\ldots\}$, which in the coordinate representation are identified with the Hermite functions. We denote by $a^*$, $a$ and $N=a^*a$ the creation, annihilation and number operators related to this basis. Moreover, $I_\hi$ or briefly $I$ is the identity operator of $\hi$.
We denote by $\lh$ and $\th$ the sets of bounded and trace class operators on $\hil$. The states of the field are represented by positive trace class operators $\rho\in\th$ with unit trace (i.e.,\ density operators), and the observables are represented by normalized positive operator valued measures (POVMs) $\E:\mathcal{B}(\Omega)\to\lh$ on the Borel $\sigma$-algebra of the topological space $\Omega$ of possible measurement outcomes.\footnote{Note that the domain of definition of a POVM can be a more general $\sigma$-algebra $\h A$ consisting of subsets of the measurement outcome set $\Omega$.} We say that an observable is sharp if it is projection valued, that is, a spectral measure. The measurement outcome probabilities are  given by the probability measures $X\mapsto \tr{\rho\E(X)}$. For our purposes, the relevant measurement outcome spaces are the torus $\T$, which is identified with $[0,2\pi)$ with addition modulo $2\pi$, the space of nonnegative real numbers $\R_+$, and the phase space $\C\simeq \R^2\simeq \R_+\times \T$.

While observables represent the statistical description of a measurement, they say nothing about the change in the system's state caused by the measurement. The effect of the measurement on the system is captured in the concept of an instrument \cite{DaLe}. Recall that a completely positive instrument is a map $\h I:\h B(\Omega)\to \h L(\tc)$, $X\mapsto \h I_X$ such that each $\h I_X:\tc\to\tc$ is completely positive, $\tr{\h I_X (\rho)} = 1$ for all states $\rho$, and the map $X\mapsto \h I_X (\rho)$ is $\sigma$-additive for all states $\rho$. Each instrument determines the dual instrument  $\h M:\h B(\Omega)\to \h L(\lh)$ by $\tr{\rho\M_X(B)}=\tr{\I_X(\rho)B}$ for $X\in\h B(\Omega)$, $\rho\in\tc$ and $B\in\lh$, and  the associated observable $\E$ via $
\tr{\rho \E(X)} = \tr{\h I_X(\rho)}$ , or equivalently, as $\E(X) = \M_X(I)$. In this case we say that the instrument (or its dual) is {\em $\E$-compatible}.  It is well known that, for a given POVM $\mathsf E$, there are infinitely many $\E$-compatible instruments, that is, an observable can be measured in many ways.

Any completely positive instrument can also be dilated to a quantum measuring process \cite{Oz84}, that is, any instrument can be realized as a unitary measurement model of the corresponding POVM: For any instrument $\mathcal I$ there exist a Hilbert space $\h K$ associated with the probe system, an initial pure state $|\xi\>\<\xi|$ of the probe, a unitary measurement interaction $U$ between the system and the probe, and a spectral measure $\P$ of the probe, the pointer observable, which define the instrument $\mathcal I$  uniquely via
$$
\h I_X (\rho) = \textrm{tr}_{\h K}\left[ U(\rho \otimes\vert \xi\rangle\langle\xi \vert)U^*I\otimes \sfq(X)\right]
$$
where $ \textrm{tr}_{\h K}[\cdot]$ denotes partial trace over the Hilbert space of the probe. Since for each POVM $\mathsf E$, there exists an infinite number of instruments which are compatible with $\mathsf E$, and any instrument has infinitely many realizations, it is natural to ask which are the minimal measurement models of $\mathsf E$. In such a measurement model there would be no unnecessary degrees of freedom in the measuring process. This question was solved in \cite{Pe13b}.

We may now turn our attention to the problem of covariant phase observables. As already mentioned before, covariance here is with respect to phase shifts which then  reflects the requirement that any phase shift performed on the field prior to the measurement should cause a corresponding shift in  the distribution of the measurement outcomes. 
\begin{definition}\label{def:covariant_phase_observable}
An observable $\E:\mathcal{B}([0,2\pi))\to\lh$ is a {\em covariant phase observable} if 
\begin{equation}\label{eqn:covariant_phase_observable}
e^{i\theta N} \E(\Theta ) e^{-i\theta N}  = \E(\Theta  \, \dot{+}\,  \theta)
\end{equation}
for all $\Theta\in\mathcal{B}([0,2\pi))$ and $\theta\in[0,2\pi)$, where $\dot{+}$ denotes addition modulo $2\pi$.
\end{definition}
Condition  \eqref{eqn:covariant_phase_observable} already tells us quite a lot about the structure of the observables. Indeed, the following structure theorem for phase observables is well-known, see e.g.\ \cite{LaPe99}.
\begin{theorem}\label{thr:structure_of_phase_observables}
Let $\E:\mathcal{B}([0,2\pi))\to\lh$ be a covariant phase observable. Then there exists a unique positive semidefinite matrix $(c_{mn})_{m,n\in\N}$ with the unit diagonal (i.e.\ $c_{mm}\equiv1$) such that 
\begin{equation}
\E(\Theta) = \sum_{m,n=0}^\infty c_{mn} \, \frac{1}{2\pi}\int_\Theta e^{i\theta(m-n)} \, d\theta \, \vert m \rangle\langle n\vert
\end{equation}
for all $\Theta\in\mathcal{B}([0,2\pi))$. We say that $(c_{mn})_{m,n\in\N}$ is the {\em phase matrix} associated to $\E$.
\end{theorem}

Let $\E$ be a phase observable with the phase matrix 
$(c_{mn})_{m,n\in\N}$. It is a standard result (i.e.\ a minimal Kolmogorov decomposition of a positive semidefinite matrix) that one can write, for all $m,\,n\in\N$,
$$
c_{mn}=\<\eta_m|\eta_n\>
$$
where the unit vectors $\eta_m\in\ki\subseteq\hi$ span a Hilbert subspace $\ki$ of $\hi$.\footnote{Recall that $\ki$ and vectors $\eta_m$ are unique up to a unitary map, i.e.\ if $\{\eta'_m\}\subseteq\ki'$ is another family of vectors satisfying the above requirements then there exists a unitary map $U:\,\ki\to\ki'$ such that $\eta_m'=U\eta_m$, $m\in\N$.}
We say that the $\eta_m$'s are (minimal) {\it structure vectors} of $\E$.

From Theorem 1 of \cite{Pe13b} follows that any $\E$-compatible completely positive dual instrument is of the form 
$$
\M_\Theta(B) = \sum_{m,n=0}^\infty  \, \frac{1}{2\pi}\int_\Theta \<\eta_m|T_\theta(B)\eta_n\>e^{i\theta(m-n)} \, d\theta \, \vert m \rangle\langle n\vert
$$
for all $\Theta\in\mathcal{B}([0,2\pi))$ and $B\in\lh$, where any $T_\theta:\,\lh\to\lk$ is a completely positive dual channel.\footnote{Recall that $\theta\mapsto\<\eta|T_\theta(B)\eta\>$ is measurable for all $\eta\in\lk$, $B\in\lh$, and the family $\{T_\theta\}_{\theta\in[0,2\pi)}$ is unique (almost everywhere).}
Indeed, by using the (minimal) Kraus decompositions
$$
T_\theta(B)=\sum_{k=1}^{r(\theta)}\A_k(\theta)^*B\A_k(\theta)
$$
where $\A_k(\theta):\,\ki\to\hi$ are bounded operators (such that $\sum_k\A_k(\theta)^*\A_k(\theta)=I_\ki$) and $r(\theta)\in\N\cup\{\infty\}$ is the rank of the channel $T_\theta$, one sees that
$$
B\mapsto\M_\Theta(B)=\int_\Theta \sum_{k=1}^{r(\theta)}\A'_k(\theta)^*B\A'_k(\theta)\frac{d\theta}{2\pi}
$$
is completely positive; here $\A'_k(\theta):=\A_k(\theta)\sum_{n}e^{-i\theta n}|{\eta_n}\>\<n|$ is a (minimal) pointwise Kraus operator of $\M$ \cite{Pe13a}. Moreover,
$\M_\Theta(I_\hi)=\E(\Theta)$ so that $\M$ is $\E$-compatible.
We say that $\M$ is {\it rank-1} if $r(\theta)\equiv1$.

Following
Davies \cite{Da}, we say that an $\E$-combatible instrument $\M$ is {\it (phase shift) covariant} if 
$$
 e^{i\theta N}\M_\Theta(e^{-i\theta N}B e^{i\theta N}) e^{-i\theta N}\equiv
\M_{\Theta\dot{+}\theta}(B)\quad\Longleftrightarrow\quad
e^{i\theta N}\I_\Theta(e^{-i\theta N}\rho e^{i\theta N}) e^{-i\theta N}\equiv\I_{\Theta\dot{+}\theta}(\rho).
$$
It is easy to check that this holds exactly when 
$$
T_\theta(B)\equiv T(e^{-i\theta N}Be^{i\theta N})
$$
where $T:\,\lh\to\lk$ is a completely positive channel. By using a (minimal) Kraus decomposition of $T$, i.e.\ $T(B)=\sum_{k=1}^{r}\A_k^*B\A_k$,
one sees that $T_\theta(B)=\sum_{k=1}^{r}\A_k(\theta)^*B\A_k(\theta)$
where $\A_k(\theta):=e^{i\theta N}\A_k$ for (almost) all $\theta$.

\section{Minimal pure measurement models of covariant phase observables}

\subsection{General case}

We will next consider {\it minimal} (pure) measurement models of phase observables. In such models, the ancillary or probe Hilbert space is the `smallest' possible.
It is quite obvious that any minimal measurement model of $\E$ is given by a rank-1 instrument, that is, $\M$ is of the form
$$
\M_\Theta(B) = \sum_{m,n=0}^\infty  \, \frac{1}{2\pi}\int_\Theta \<\eta_m|\A(\theta)^*B\A(\theta)\eta_n\>e^{i\theta(m-n)} \, d\theta \, \vert m \rangle\langle n\vert
$$
where $\A(\theta):\,\ki\to\hi$ is an isometry for (almost) all $\theta$, see \cite[Section 5]{Pe13b}.
Now the ancillary space is (unitarily equivalent with) $\hi'=L^2([0,2\pi))$, consisting of square integrable `wave functions' $\psi:\,[0,2\pi)\to\C$, and the pointer observable is the canonical spectral measure $\P:\,\mathcal{B}([0,2\pi))\to {\mathcal L}(\hi')$,
$$
(\P(\Theta)\psi)(\theta):=\chi_\Theta(\theta)\psi(\theta)
$$
where $\chi_\Theta$ is the characteristic function of $\Theta$ \cite[Remark 3]{Pe13b}.
Since the vectors $e_n(\theta):=e^{-i\theta n}/\sqrt{2\pi}$, $n\in\Z$, form an orthonormal basis of $\hi'$ one can write
$$
\P(\Theta) = \sum_{m,n=-\infty}^\infty\<e_m|\P(\Theta)e_n\>\vert e_m \rangle\langle e_n\vert
=\sum_{m,n=-\infty}^\infty\frac{1}{2\pi}\int_\Theta e^{i\theta(m-n)} \, d\theta \, \vert e_m \rangle\langle e_n\vert.
$$
If the initial pure state $\xi\in\hi'$, $\|\xi\|=1$, of the apparatus is given, then one can construct a unitary map $U:\,\hi\otimes\hi'\to\hi\otimes\hi'$ by extending
$$
[U(|n\>\otimes\xi)](\theta)=(\A(\theta)\eta_n)\otimes e_n(\theta).
$$
Indeed, for all $m,\,n\in\N$,
\begin{eqnarray*}
\tr{U(|n\>\<m|\otimes|\xi\>\<\xi|)U^*(I_\hi\otimes\P(\Theta))}&=&
\<U(|m\>\otimes\xi)|(I_\hi\otimes\P(\Theta))U(|n\>\otimes\xi)\> \\
&=&\int_{\Theta}\<{[U(|m\>\otimes\xi)](\theta)}|[U(|n\>\otimes\xi)](\theta)\>d\theta \\
&=&\int_{\Theta}\<\A(\theta)\eta_m|\A(\theta)\eta_n\>e^{i\theta(m-n)}\frac{d\theta}{2\pi} \\
&=&
\int_{\Theta}\<\eta_m|\eta_n\>e^{i\theta(m-n)}\frac{d\theta}{2\pi}=\tr{|n\>\<m|\E(\Theta)}
\end{eqnarray*}
so that, for any state $\rho\in\th$, the phase measurement probability distribution $\Theta\mapsto\tr{\rho\E(\Theta)}$ can be obtained
by measuring the pointer $\P$ in the state $U(\rho\otimes|\xi\>\<\xi|)U^*$ where $\xi$ is the initial pure state of the apparatus and $U$ describes the measurement interaction between the system and probe.

If, additionally, we assume that the rank-1 instrument (or the minimal measurement) is phase shift covariant then one must have
$\A(\theta)=e^{i\theta N}\A$
where $\A:\,\ki\to\hi$ is an isometry.
In this case,
$$
\M_\Theta(B) = \M^\A_\Theta(B) := \sum_{m,n=0}^\infty  \, \frac{1}{2\pi}\<\A\eta_m|e^{-i\theta N}Be^{i\theta N}\A\eta_n\>\int_\Theta e^{i\theta(m-n)} \, d\theta \, \vert m \rangle\langle n\vert.
$$
Since $c_{mn}=\<\eta_m|\eta_n\>=\<\A\eta_m|\A\eta_n\>$ the vectors $\A\eta_m$ are new structure vectors of $\E$ and it is enough to study the case
$\A={\bm1}$ where $\bm1\eta:=\eta$ for all $\eta\in\ki$. 
The predual instrument of $\M^{\bm1}$ is
$$
\I_\Theta^{\bm1}(\rho)=\sum_{m,n=0}^\infty  \<n|\rho|m\> \frac{1}{2\pi}\int_\Theta   e^{i\theta N}|\eta_n\>\<\eta_m|e^{-i\theta N}e^{i\theta(m-n)} d\theta
$$
and it describes the minimal covariant measurement of the phase observable $\E$ with structure vectors $\eta_m$.

\begin{remark}\rm\label{remarkki1}
We have seen that the minimal ancillary space of any phase measurement is (unitarily equivalent with) $L^2([0,2\pi))$ and the pointer observable can be identified with the canonical spectral measure $\P$. However, the physical interpretation of $\P$ in the case of a single-mode optical field is more or less unclear.

If, instead of the minimal  space $L^2 ([0,2\pi))$, we consider an ancillary space consisting of two modes with the lowering operators $a\otimes I_\hi$ and $I_\hi\otimes a$ (operating on the two-mode Hilbert space $\hi\otimes\hi$) then $\P$ could be viewed as a spectral measure acting on a subspace of  $\hi\otimes\hi$ in the following way:
Choose an isometry $J:\,L^2([0,2\pi))\to\hi\otimes\hi$. Thus, $\hi'':=JL^2([0,2\pi))\subseteq\hi\otimes\hi$ is a subspace and the vectors $|m\>\>:=Je_m$, $m\in\Z$, form its orthonomal basis. Now
$$
J\P(\Theta)J^* = \sum_{m,n=-\infty}^\infty \, \frac{1}{2\pi}\int_\Theta e^{i\theta(m-n)} \, d\theta \, \vert m \>\>\<\< n\vert
$$
is a spectral measure on the minimal ancillary space $\hi''$.

For example, one may choose $|m\>\>=|m,0\>\>$ where, for all $m\in\Z,\;k\in\N$,
$$
|m,k\>\>:=
\begin{cases}
|k+m\>\otimes|k\>, & m\ge 0,\\
|k\>\otimes |k-m\>, & m<0,
\end{cases}
$$
are Ban's relative number states (RNSs) \cite{Ban1} which are the eigenstates of the 
number difference $\Delta N=N\otimes I - I\otimes N$ and the total number
$\Sigma N=N\otimes I + I\otimes N $ operators:
$$
\Delta N |m,k\>\>= m |m,k\>\>,\qquad
\Sigma N |m,k\>\>= (2k+|m|) |m,k\>\>.
$$
Moreover, they form an orthonormal basis of $\hi\otimes\hi$ since
\begin{equation}
\label{Banban}
|m\>\otimes|n\>=|m-n,\min\{m,n\}\>\>,\qquad m,n\in\N.
\end{equation}
Now $J\P J^*$ is the restriction of the spectral measure $\mathsf O_{\rm angle}:\,\mathcal B([0,2\pi))\to\mathcal L(\hi\otimes\hi)$,
\begin{equation}
\label{Qpilkku}
\mathsf O_{\rm angle}(\Theta):=\sum_{m,n=-\infty}^\infty \, \frac{1}{2\pi}\int_\Theta e^{i\theta(m-n)} \, d\theta \sum_{k=0}^\infty\, \vert m,k \>\>\<\< n,k\vert
\end{equation}
to $\hi''$ (see, Remark \ref{remu2}). Note that $\tr{(\rho\otimes|0\>\<0|)\mathsf O_{\rm angle}(\Theta)}=\tr{\rho\Ecan(\Theta)}$ where $\Ecan$ is the canonical phase.
\end{remark}

\subsection{Canonical phase}
We will now consider the canonical phase observable $\Ecan$ which is determined by the constant phase matrix $c_{mn}\equiv1$ or by structure vectors $\eta_m\equiv\eta$ where $\eta\in\hi$ is any fixed unit vector, that is,
$$
\Ecan(\Theta) = \sum_{m,n=0}^\infty \frac{1}{2\pi}\int_\Theta e^{i\theta(m-n)} \, d\theta \, \vert m \rangle\langle n\vert.
$$
Now $\ki=\C\eta\simeq\C$ and the Kraus operators $\A(\theta):\,\C\eta\to\hi$ (related to a minimal measurement) can be identified with the unit vectors $\psi_\theta:=\A(\theta)\eta\in\hi$.  
Now the corresponding instrument is 
$$
\M_\Theta(B) = \sum_{m,n=0}^\infty  \, \frac{1}{2\pi}\int_\Theta \<\psi_\theta|B\psi_\theta\>e^{i\theta(m-n)} \, d\theta \, \vert m \rangle\langle n\vert
$$
and 
$$
[U(|n\>\otimes\xi)](\theta)=\psi_\theta\otimes e^{-i\theta n}.
$$
Especially, the phase shift covariant (minimal) instruments are of the form\footnote{Recall that any $\Ecan$-compatible instrument is nuclear, i.e.\ of the form
$
\I_\Theta(\rho)=\int_\Theta\sigma_\theta\,\tr{\rho\Ecan(d\theta)}
$
where $\{\sigma_\theta\}_{\theta\in[0,2\pi)}\subseteq\th$ is a (measurable) family of states \cite[Theorem 2]{Pe13b}. Now $\sigma_\theta$ can be interpreted as a posterior state when the outcome $\theta$ is obtained. It does not depend on the initial state $\rho$.
In the case of a minimal covariant measurement of the canonical phase, $\sigma_\theta=e^{i\theta N}|\eta\>\<\eta|e^{-i\theta N}$ for some unit vector $\eta$. The choices $\eta=|n\>$, $n\in\N$, give the only rotation invariant states
$\sigma_\theta\equiv|n\>\<n|$.
}
$$
\I_\Theta^{\bm1}(\rho)=\sum_{m,n=0}^\infty  \<n|\rho|m\> \frac{1}{2\pi}\int_\Theta   e^{i\theta N}|\eta\>\<\eta|e^{-i\theta N}e^{i\theta(m-n)} d\theta
=\int_\Theta e^{i\theta N}|\eta\>\<\eta|e^{-i\theta N}\,\tr{\rho\Ecan(d\theta)}
$$
and the related measurement interactions are given by the formula 
$[U(|n\>\otimes\xi)](\theta)=(e^{i\theta N}\eta)\otimes e^{-i\theta n}$ which can be extended to a unitary operator in infinitely many ways.
Especially, if one chooses $\eta=|0\>$ and $\xi=e_0$ then
$U(|n\>\otimes e_0)=|0\>\otimes e_n$, $n\in\N$, which extends, e.g., to a unitary operator $U:\,\hi\otimes\hi'\to\hi\otimes\hi'$,
$$
U=\sum_{k=-\infty}^{-1}\sum_{n=0}^\infty|n\>\<n|\otimes|e_k\>\<e_k|
+\sum_{k=0}^{\infty}\sum_{n=0}^\infty|k\>\<n|\otimes|e_n\>\<e_k|
$$
whose eigenvectors are $|n\>\otimes e_k$, $n\ge0$, $k<0$ (related to the eigenvalue $e^{i0}=1$) and, for all $k,\,n\ge 0$, 
\begin{eqnarray*}
\psi_{nk}^\pm&=&
2^{-1/2}(|n\>\otimes e_k\pm|k\>\otimes e_n),\qquad k\ne n, \\
\psi_{nn}^+&=&|n\>\otimes e_n,
\end{eqnarray*}
which are related to the eigenvalues $+1=e^{i0}$ and $-1=e^{i\pi}$.
Hence, $U=e^{iH}$ where the generator (Hamiltonian) $H$ of the measurement interaction is
$$
H=\pi\sum_{k,n\in\N\atop k\ne n}|\psi_{nk}^-\>\<\psi_{nk}^-|.
$$
By using the isometry $J$ of Remark \ref{remarkki1}, $Je_m=|m\>\>=|m,0\>\>=
|m\>\otimes|0\>\in\hi\otimes\hi$  when $m\ge 0$, one sees that
$$
(I_\hi\otimes J)\psi_{nk}^-=2^{-1/2}(|n\>\otimes Je_k-|k\>\otimes Je_n)
=\vp^-_{nk}\otimes|0\>\in\hi\otimes\hi\otimes\hi.
$$
where $\vp^-_{nk}=2^{-1/2}(|n\>\otimes |k\>-|k\>\otimes|n\>)$
so that 
$$
H_J:=
(I_\hi\otimes J)H(I_\hi\otimes J)^*
=\pi\sum_{k,n\in\N\atop k\ne n}|\vp_{nk}^-\>\<\vp_{nk}^-|\otimes|0\>\<0|.
$$
Since $J^*J=I_{\hi'}$,
$$
U_J:=(I_\hi\otimes J)U(I_\hi\otimes J)^*=(I_\hi\otimes J)e^{iH}(I_\hi\otimes J)^*=e^{iH_J}
$$ 
and $U_J$ just flips the vectors of the first two modes:
$$
U_J(\psi\otimes\vp\otimes|0\>)=\vp\otimes\psi\otimes|0\>,\qquad \psi,\,\vp\in\hi.
$$ 
However, it is still an open question how to measure the pointer observable
$$
J\P(\Theta)J^* = \sum_{m,n=-\infty}^\infty \, \frac{1}{2\pi}\int_\Theta e^{i\theta(m-n)} \, d\theta \, \vert m \>\>\<\< n\vert.
$$
We will come back to this problem in Remark \ref{remu2} below.

\section{Joint measurements including the canonical phase}

As we saw in the preceding section, to measure a phase observable $\E$, one needs a pointer observable $\P_a:\,\mathcal B(\Omega)\to\mathcal L(\hi_a)$, $\Omega\subseteq\R^n$, which satisfies the following requirements:
a) the Hilbert space $\hi_a$ of the apparatus must be `larger' than $L^2[0,2\pi)$, i.e.\ there exists an isometry $J:\,L^2[0,2\pi)\to\hi_a$, and
b) there exists a pointer function (measurable surjection) $f:\,\Omega\to[0,2\pi)$ such that
\begin{equation}\label{pointteri}
J^*\P_a(f^{-1}(\Theta))J\equiv\P(\Theta).
\end{equation}
For example, if one chooses $\Omega=\C\simeq\R^2$ and uses polar coordinates $(r,\theta)$ then the choices $\hi_a=L^2(\C)\simeq L^2(\R_+)\otimes L^2[0,2\pi)$,
 $f(r,\theta):=\theta$, $(\P_a(Z)\Psi)(z)=\chi_Z(z)\Psi(z)$, $Z\in\mathcal B(\C)$, $\Psi\in\hi_a$, $z\in\C$,  and $(Je_n)(r,\theta)=g(r)e_n(\theta)$, $\int_0^\infty|g(r)|^2dr=1$ satisfy \eqref{pointteri}.
Similarly, one can choose $\Omega=\R^3$ (or any $\R^n$) and change the cartesian coordinates to (generalized) spherical coordinates. Now one of the coordinates is $\theta\in[0,2\pi)$ and we may generalize easily the above construction.
Hence,
an interesting special case of \eqref{pointteri} is the case of $\Omega=\Omega'\times[0,2\pi)$ where $\Omega'\subseteq\R^{n-1}$ and $f(\omega,\theta)=\theta$ for all $\theta\in[0,2\pi)$ and $\omega\in\Omega'$.
Now $\P_a$ can be viewed as a joint measurement of the POVMs $\Theta\mapsto \P_a(\Omega'\times\Theta))$ and $\mathcal B(\Omega')\ni X\mapsto \P_a(X\times[0,2\pi))$.
Note that if $U$ is the measurement interaction and $\xi$ the initial pure state of the apparatus then
\begin{equation} \label{kulmakaavapointteri}
\E(\Theta)=V_\xi^*U^*(I_\hi\otimes\P_a(\Omega'\times\Theta))UV_\xi
\end{equation}
where $V_\xi:\,\hi\to\hi\otimes\hi_a$, $\psi\mapsto \psi\otimes\xi$, is an isometry.
By defining another POVM 
$$
\F(X):=V_\xi^*U^*(I_\hi\otimes\P_a(X\times[0,2\pi))UV_\xi
$$
we see that $\E$ and $\F$ are jointly measurable and $\P_a$ (resp.\ $U$) can be used as a pointer observable (resp.\ measurement interaction) also for $\F$.
Next we give the definition for jointly measurable observables and study the case of the canonical phase.

\begin{definition}
Let $\E_1:\h B(\Omega_1)\to \lh$ and $\E_2:\h B(\Omega_2) \to \lh$ be observables.  We say that $\E_1$ and $\E_2$ are {\em jointly measurable} if there exists an observable $\E:\h B(\Omega_1\times \Omega_2)\to\lh$   such that 
\begin{equation}
\E_1(X) = \E(X\times \Omega_2),\qquad \E_2(Y) = \E(\Omega_1 \times Y),
\end{equation}
for all $X\in\h B(\Omega_1)$ and $Y\in \h B(\Omega_2)$.
\end{definition}

It is a standard result that for sharp observables joint measurability is equivalent to commutativity. The same is true also in the case that one (but not both) of the observables is merely a POVM. 
In such a case, $\E(X\times Y)\equiv\E_1(X)\E_2(Y)\equiv\E_2(Y)\E_1(X)$.
We will see (Prop.\ \ref{prop:sharp_joint_measurement} below) that in the case of the canonical phase this rules out any joint measurements where the other observable is sharp. This is a consequence of the following lemma.

\begin{lemma}\label{lemma:commutator}
Let $A\in\lh$. Then $[A,\E_\mathrm{can}(\Theta)]=0$ for all $\Theta\in\h B([0,2\pi))$ if and only if $A=cI$ for some $c\in \C$.
\end{lemma}
\begin{proof}
First notice that $A$ can be decomposed as $A= A_1 +  i A_2$ where $A_j\in\lh$ are selfadjoint; for example, $A_1=(A+A^*)/2$. Therefore, if $A$ commutes with all $\E_\mathrm{can}(X)$, then by $[A^*,\E_\mathrm{can}(X)]= -[A,\E_\mathrm{can}(X)]^*$, so does $A^*$ and it follows that both $A_j$:s also commute with all $\E_\mathrm{can}(X)$.  Now for any $k\in\Z$ define the operator 
$$
B_k= \frac{1}{2\pi} \int e^{ik\theta} \, d\E_\mathrm{can} (\theta) = \sum_{m,n=0}^\infty \frac{1}{2\pi}\int e^{i(k+m-n)\theta} \, d\theta \, \vert m\rangle\langle n\vert = \sum_{m = \max \{ 0,-k\}}^\infty \vert m \rangle\langle m+k\vert. 
$$
If $A_j$ commutes with $\E_\mathrm{can}$, it also commutes with all $B_k$:s. Thus, we have $\langle n \vert [A_j,B_k] \vert l\rangle = 0$ for all $n,l\in\N$ and $k\in\Z$. In particular, if $k>l$ then 
$$
\langle 0 \vert [A_j,B_k]\vert l \rangle = - \langle k \vert A_j\vert l \rangle  = 0
$$
so that $A_j$ must be diagonal. (Recall that $A_j$ is  selfadjoint.) But then we also have
$$
\langle 0 \vert [A_j,B_k]\vert k \rangle  = \langle 0 \vert A_j \vert 0\rangle  - \langle k\vert A_j\vert k \rangle = 0
$$
for all $k\geq 0 $ so that $A_j=c_jI$ for some $c_j\in\R$. Thus, $A= cI$ where $c=c_1 + ic_2\in\C$.
\end{proof}

\begin{proposition}\label{prop:sharp_joint_measurement}
The canonical phase observable is not jointly measurable with any nontrivial sharp observable.
\end{proposition}
\begin{proof}
If, say, $\F:\h B(\Omega)\to \lh$ is a sharp observable which is jointly measurable with the canonical phase $\E_\mathrm{can}$, then $\F$ and $\E_\mathrm{can}$ must commute according to \cite[Theorem 1.3.1]{Ludwig}.  By Lemma \ref{lemma:commutator} (and the fact that $\F$ is an observable), for each $X\in \h B(\Omega)$ there exists a $c_X\in[0,1]$ such that $\F(X) = c_X I$.  Thus, $\F$ must be trivial. (Moreover, since $\F$ is projection valued, one must have $(c_X)^2=c_X$ so that $X\mapsto c_X\in\{0,1\}$ is an extremal probability measure which is usually a Dirac measure).
\end{proof}

The next proposition shows that even if we allow the other observable to be a POVM, we still do not gain anything by performing a joint measurement. Indeed, we see that the other observable is necessarily a post-processing of $\E_\mathrm{can}$ and therefore it does not give any more information than the measurement of merely $\E_{\mathrm{can}}$.

\begin{proposition}
Let $\F:\h B(\Omega)\to\lh$ be an observable. Then $\F$ is jointly measurable with the canonical phase observable $\E_\mathrm{can}$ if and only if there exists a (weak) Markov kernel $m: \, \h B(\Omega)\times[0,2\pi) \to [0,1]$ such that  
\begin{equation}\label{eqn:post_proc}
\F(X) = \int_0^{2\pi} m(X,\theta)\, d \E_{\mathrm{can}} (\theta)
\end{equation}  
for all $X\in \h B(\Omega)$. In such a case, the joint observable $\mathsf{M}: \h B(\Omega \times [0,2 \pi))\to\lh$ is unique and is given by 
\begin{equation}\label{eqn:joint_observable}
\mathsf{M}(X\times \Theta) = \sum_{m,n=0}^\infty \, \frac{1}{2\pi}\int_\Theta m(X,\theta)\, e^{i\theta(m-n)} \, d\theta \, \vert m \rangle\langle n\vert
\end{equation}
for all $X\in \h B(\Omega)$ and $\Theta\in\h B([0,2\pi))$.
\end{proposition}
\begin{proof}
Assume first that $\F$ and $\E_\mathrm{can}$ are jointly measurable, with the joint observable $\mathsf{M}$, i.e., $\mathsf{M}_1 = \F$ and $\mathsf{M}_2=\E_\mathrm{can}$. Consider the minimal Naimark dilation $(\h K, \mathsf{Q}, V)$ of  $\E_\mathrm{can}$ where $\h K = L^2 ([0,2\pi))$, $\mathsf{Q}$ is the canonical spectral measure on $\h K$, and  $V:\hil\to\h K$ is the isometry $V\vert n\rangle =e_n$ where $e_n(\theta) = \frac{1}{\sqrt{2\pi}} e^{-in\theta}$.  Since $\mathsf{M}(X\times \Theta)\leq \mathsf{M} (\Omega\times \Theta)= \E_\mathrm{can}(\Theta)$ for all $X\in\h B(\Omega)$ and $\Theta\in\h B([0,2\pi))$, there exists an observable $\mathsf{R}: \h B(\Omega)\to \h L(\h K)$  which commutes with $\mathsf{Q}$ and  
\begin{equation}\label{eqn:joint_dilation}
\mathsf{M}(X\times \Theta)  =  V^* \mathsf{R}(X) \mathsf{Q}(\Theta) V
\end{equation}
for all  $X\in \h B(\Omega)$ and $\Theta\in\h B([0,2\pi))$ \cite[Lemma 4.1]{HaHePe}. But the canonical spectral measure $\mathsf{Q}$ is maximal in the sense that any operator which commutes with $\mathsf{Q}$ must be a function of it \cite[Theorem 1, p.\ 187]{Dix}.
 Indeed, for a fixed $X\in \h B(\Omega)$ and for any $\Theta\in\h B([0,2\pi))$ we have 
$$
\left[\mathsf{R}(X)\chi_\Theta\right](\theta) = \left[\mathsf{R}(X) \mathsf{Q}(\Theta) 1\right](\theta) = \left[\mathsf{Q}(\Theta)\mathsf{R}(X) 1\right](\theta) = \left[\mathsf{R}(X) 1\right](\theta)\chi_\Theta (\theta)
$$
for almost all $\theta\in[0,2\pi)$ where $1$ denotes the constant function $\theta\mapsto 1$. Since the linear combinations of characteristic functions are dense in $\h K$, we conclude that $\mathsf{R}(X)$ corresponds to multiplication by the function $\mathsf{R}(X) 1$. By denoting $m(X,\theta)=(\mathsf{R}(X)1)(\theta)$ we obtain a (weak) Markov kernel $m:\h B(\Omega)\times [0,2\pi) \to [0,1]$, and it follows from Eq.\ \eqref{eqn:joint_dilation} that 
$$
\F(\Theta) = V^* \mathsf{R}(X) V = \int_0^{2\pi} m(X,\theta) \, d\E_\mathrm{can}(\theta). 
$$

Suppose now that there exists a (weak) Markov kernel such that Eq. \eqref{eqn:post_proc} holds. Since for each $X\in\h B(\Omega)$ the map $\theta\mapsto m(X,\theta)$ is measurable and $m(X,\theta)\leq 1$ for almost all $\theta$, we may define a bounded operator $\mathsf{R}(X)\in\h L(\h K)$ by $(\mathsf{R}(X)\varphi)(\theta) = m(X,\theta) \varphi (\theta)$. The map $X\mapsto \mathsf{R}(X)$ is then a POVM which commutes with $\sfq$. Therefore the map $(X,\Theta)\mapsto V^* \mathsf{R}(X) \mathsf{Q}(\Theta) V$ extends to an observable $\mathsf{M}:\h B(\Omega \times [0,2\pi)) \to\lh$ whose margins are $\mathsf{F}$ and $\E_\mathrm{can}$. In other words, $\mathsf{F}$ and $\E_\mathrm{can}$ are jointly measurable. 

In both of the above instances the joint observable $\mathsf{M}$ satisfies Eq.\ \eqref{eqn:joint_observable}, and the uniqueness follows from \cite[Theorem 4.1(a)]{HaHePe} since $\E_\mathrm{can}$ is an extreme point of the convex set of all observables on $\h B([0,2\pi))$ \cite{HePe}.
\end{proof}

Note that, by combining the preceding two propositions, we see that $\mathsf F$ is projection valued if and only if the (weak) Markov kernel $m(X,\theta)=c_X=(c_X)^2$ (almost all $\theta$), i.e.\ it does not depend on $\theta$  and $c_X\in\{0,1\}$.

\section{Phase shift covariant phase space observables}\label{sec:PSCPSO}
Even though we have already seen that no additional information can be obtained by measuring some other observable together with the canonical phase, there is a strong practical motivation for studying such joint measurement. For instance, the  realization of the measurement of the canonical spectral measure $\sfq:\h B([0,2\pi))\to L^2 ([0,2\pi))$ may be problematic, but things get much simpler if we instead consider $\sfq$ as the angle margin of the canonical spectral measure $\mathsf{M}:\h B(\C)\to L^2(\C)\simeq  L^2(\R_+) \otimes L^2 ([0,2\pi))$, i.e., we interpret $\sfq $ as the observable $X\mapsto I\otimes \sfq (X)$ on the larger Hilbert space. In this case $L^2(\C)$ may be simply realized as the Hilbert space of a two-mode field, whereas $\mathsf{M}$, and therefore $\sfq$, can be measured by performing homodyne detection on the two modes.

For any observable $\mathsf{P}:\mathcal{B}(\C)\to\lh$ we define the angle and radial margins $\mathsf{P}_\mathrm{angle}:\mathcal{B}([0,2\pi))\to\lh$ and $\mathsf{P}_\mathrm{rad}:\mathcal{B}(\R_+)\to\lh$ via 
\begin{equation}\label{eqn:margins}
\mathsf{P}_\mathrm{angle}(\Theta) = \mathsf{P}(\R_+\times \Theta),\qquad \mathsf{P}_\mathrm{rad}(X) = \mathsf{P}(X\times[0,2\pi)).
\end{equation}
Any phase observable $\E$ can obviously be obtained as the angle margin of a phase space observable in a trivial manner. Namely, for any probability measure $\mu:\mathcal{B}(\R_+)\to[0,1]$, the observable $\mathsf{P}:\mathcal{B}(\C)\to\lh$ for which $\mathsf{P}(X\times \Theta)= \mu(X)\E(\Theta)$ is a suitable choice. This observable has the additional property that it also transforms covariantly under the action of the phase shifter, i.e., 
$$
e^{i\theta N}\mathsf P(X\times\Theta)e^{-i\theta N}= \mu(X)e^{i\theta N}\E(\Theta) e^{-i\theta N} =  \mu(X)\E(\Theta\dot{+} \theta)= \mathsf P\big(X\times(\Theta\dot{+}\theta)
\big).
$$
We will now focus our attention on phase space observables which have this symmetry property.

\begin{definition}\label{def:PSCPSO}
A {\em phase shift covariant phase space observable} is  a POVM $\mathsf P:\,\mathcal B(\mathbb C)\to\mathcal L(\mathcal H)$ such that 
$$
e^{i\theta N}\mathsf P(X\times\Theta)e^{-i\theta N}=\mathsf P\big(X\times(\Theta\dot{+}\theta)
\big)
$$
for all $X\in\mathcal B(\mathbb R_+)$, $\Theta\in\mathcal B([0,2\pi))$, and $\theta\in [0,2\pi)$.
\end{definition}

The following structure theorem for phase shift covariant phase space observables can then be proved.

\begin{proposition}\label{prop:structure}
An observable $\mathsf P:\,\mathcal B(\mathbb C)\to\mathcal L(\mathcal H)$ is a phase shift covariant phase space observable if and only if there exists a probability measure $\mu:\,\mathcal B(\mathbb R_+)\to[0,1]$ and (a weakly $\mu$--measurable field of) vectors 
$\eta_m(x)\in\mathcal H$ satisfying the condition
$$
\int_0^\infty\|\eta_m(x)\|^2 d\mu(x)=1
$$
such that 
\begin{equation}\label{eqn:structure}
\mathsf P(X\times\Theta) = \sum_{m,n=0}^\infty\frac{1}{2\pi}\int_{X\times\Theta} \langle e^{-im\theta}\eta_m(x)|e^{-in\theta}\eta_n(x)\rangle d\theta d\mu(x)|m\rangle\langle n|
\end{equation}
for all $X\in\mathcal B(\mathbb R_+)$ and $\Theta\in\mathcal B([0,2\pi))$.
\end{proposition}
\begin{proof}
Consider a fixed set $X\in\mathcal B(\mathbb R_+)$. Since $e^{i\theta N}\mathsf P(X\times\Theta)e^{-i\theta N}=\mathsf P\big(X\times(\Theta\dot{+}\theta)
\big)$ for all $\Theta\in\mathcal B([0,2\pi))$, following the proof of the structure theorem for covariant phase observables \cite{LaPe99},  there exists a positive semidefinite matrix $\big(c_{mn}(X)\big)_{m,n\in\N}$ such that 
$$
\mathsf P(X\times\Theta)=
\sum_{m,n=0}^\infty c_{mn}(X) \, \frac{1}{2\pi}\int_\Theta
e^{i(m-n)\theta}d\theta 
|m\rangle\langle n|
$$
Now, for all $m,n\in\N$ the map $X\mapsto c_{mn}(X)$ is a complex measure 
which is clearly absolutely continuous with respect to the probability measure
$$
X\mapsto\mu(X)=\mathsf \sum_{n=0}^\infty\lambda_n \big\langle n\big|\mathsf P\big(X\times[0,2\pi)\big)\big|n\big\rangle
=\sum_{n=0}^\infty\lambda_nc_{nn}(X)
$$
where $\lambda_n>0$ for all $n$ and $\sum_n\lambda_n=1$. It follows from \cite{HyPeYl} that there exist vectors  $\eta_m(x)\in\mathcal H$ such that $c_{mn}(X)=\int_X\langle\eta_m(x)|\eta_n(x)\rangle d\mu(x)$. Hence, Eq.\ \eqref{eqn:structure} holds. The converse claim is clearly true.
\end{proof}

Using Prop.\ \ref{prop:structure} we can now determine the angle and radial margins of any phase shift covariant phase space observable. Indeed, by setting $X = \R_+$ in Eq.\ \eqref{eqn:structure} we obtain
\begin{equation}\label{eqn:angle_margin}
\mathsf{P}_\textrm{angle} ( \Theta)  = \sum_{m,n=0}^\infty c_{mn} \, \frac{1}{2\pi} \int_\Theta e^{i\theta (m-n)} \, d\theta \vert m\rangle\langle n\vert 
\end{equation}
where the phase matrix elements are given by 
$$
c_{mn} = \int_0^\infty\langle\eta_m(x)|\eta_n(x)\rangle d\mu(x).
$$ 
Similarly, the radial margin is seen to be 
\begin{equation}\label{eqn:radial_margin}
\mathsf{P}_\textrm{rad} (X) = \sum_{m=0}^\infty c_{mm}(X) \vert m\rangle \langle m \vert.
\end{equation}
The map $(X,m)\mapsto c_{mm}(X)$ is a Markov kernel  on $ \mathcal{B}(\R_+)\times \N$, and the observable $\mathsf{P}_\textrm{rad}$ can therefore be viewed as a smeared number observable. More generally, we say that an observable $\F: \h B(\Omega)\to\lh$ is a smeared number observable if there exists a  Markov kernel $m:\h B(\Omega)\times \N\to [0,1]$ such that 
$$
\F(X)  = \sum_{n=0}^\infty m(X,n) \vert n \rangle \langle n\vert
$$ 
for all $X\in\h B(\Omega)$.

\begin{proposition}\label{prop:joint_observable}
Let $\mathsf{F}:\mathcal{B}(\R_+)\to\lh$ be a smeared number observable  and let $\mathsf{E}:\mathcal{B}(\R_+)\to\lh$ be a covariant phase observable. 
 Then $\mathsf{E}$ and $\mathsf{F}$ have a joint observable if and only if they have a joint observable which is a phase shift covariant phase space observable.
\end{proposition}
\begin{proof}
Assume that $\mathsf{E}$ and $\mathsf{F}$ have a joint observable $\mathsf{M}:\mathcal{B}(\C)\to\lh$ and define the biobservable $\mathsf{P}':\mathcal{B}(\R_+)\times \mathcal{B}([0,2\pi))\to\lh$ via  
\begin{equation}\label{eqn:biobservable}
\mathsf{P}'(X,\Theta) = \frac{1}{2\pi} \int e^{-i\theta N} \mathsf{M} (X\times ( \Theta \dot{+} \theta)) e^{i\theta N} \, d\theta 
\end{equation}
where the integral is understood in the weak$^*$-sense. By \cite[Theorem 6.1.5]{LaYl2004}, $\mathsf{P}'$ extends to a unique observable $\mathsf{P}:\mathcal{B}(\C)\to\lh$ which is clearly phase shift covariant. Setting $\Theta = [0,2\pi)$ (resp. $X=\R_+$) in Eq.\ \eqref{eqn:biobservable} and using the phase shift invariance of $\mathsf{F}$ (resp.\ phase shift covariance of $\mathsf{E}$) we then see that $\mathsf{P}_\textrm{rad} = \mathsf{F}$ (resp.\ $\mathsf{P}_\textrm{angle} = \mathsf{E}$). Hence, $\mathsf{E}$ and $\mathsf{F}$ have a joint observable which is a phase shift covariant phase space observable. The converse statement is trivial.
\end{proof}

We immediately see that the canonical phase observable is obtained as the angle margin of a phase shift covariant phase space observable if and only if 
$$
\langle \eta_m | \eta_n \rangle := \int_0^\infty\langle\eta_m(x)|\eta_n(x)\rangle d\mu(x)  = 1
$$
for all $m,n\in\N$. In particular, we must have $\langle \eta_0 | \eta_n \rangle =1$ for all $n\in\N$ and since $\eta_n$ is a unit vector, this implies that $\eta_n = \eta_0$ for all $n\in\N$. But in such a case the radial margin is the trivial observable $\mathsf{P}_\textrm{rad} (X) = c_{00} (X) I $. A similar situation happens if we insist that the radial margin is the sharp number observable. In this case we have $c_{mm}(X)=\delta_m (X)$ (the Dirac measure concentrated on $m$) so that, in particular, $\int_{\{m\}} \Vert \eta_m(x)\Vert^2 \, d\mu(x) =1$ and we must have $\eta_m(x)=0$ for $\mu$-almost all $x\in \R_+ \setminus \{m\}$. It follows that $c_{mn}=\langle \eta_m | \eta_n\rangle = \delta_{mn}$ and the angle margin is therefore the trivial one: $\mathsf{P}_\textrm{angle} ( \Theta)   = \frac{1}{2\pi}\int_\Theta \, d\theta \, I$.

\begin{example}\label{ex1}\rm
As a first example, we consider phase space observables which, in addition to phase shift covariance, are covariant with respect to phase space translations represented by the displacement operators $D(z) = e^{za^* -\overline{z} a}$, $z\in\C$. Generally, any phase space POVM $\G$ which satisfies the covariance condition $D(z)\G(Z)D(z)^* = \G(Z+z)$, $z\in\C$, $Z\in\mathcal B(\C)$, is generated by a unique positive trace one operator $T$, giving the corresponding observable the explicit form \cite{Holevo1979, Werner1984}
\begin{equation}\label{eqn:ex1}
\G^T (Z) =\frac{1}{\pi} \int_Z D(z) TD(z)^* \, d^2z.
\end{equation}
However, not all of these observables are covariant with respect to phase shifts. Indeed, this is the case if and only if the generating operator is a mixture of number states, $T  = \sum_{k=0}^\infty \lambda_k \vert k \rangle\langle k\vert$, see \cite{LaPe99} for the proof of the `only if' -part and the calculation below. In order to connect these observables to the general structure discussed above, let us consider $\G^{\vert k\rangle}$, the observable generated by a single number state $\vert k \rangle\langle k\vert$. First note that we may write Eq.\ \eqref{eqn:ex1} as 
\begin{equation}\label{eqn:ex2a}
\G^{\vert k \rangle} (X\times \Theta)  = \sum_{m,n=0}^\infty \frac{1}{2\pi} \int_{X\times \Theta}  \langle m \vert D(re^{i\theta})  \vert k\rangle \overline{\langle n \vert D(re^{i\theta})  \vert k\rangle }\,  dr^2 \,d\theta\vert m \rangle\langle n \vert
\end{equation}
where the matrix elements of the displacement operators are given by
\begin{equation}\label{eqn:displacement_elements}
\langle m \vert D(re^{i\theta})  \vert k\rangle = (-1)^{\max\{ 0,k-m\}} \sqrt{\tfrac{(\min\{ m,k\})!}{(\max\{m,k\})!}} e^{i\theta(m-k)}r^{\vert m-k\vert}  L^{\vert m-k\vert}_{\min\{ m,k\}} (r^2) e^{-\frac{r^2}{2}}
\end{equation}
and $L_n^\alpha$ is the associated Laguerre polynomial (which generalizes the Laguerre polynomials $L_n=L_n^0$).
We can then define the probability measure $\mu:\h B(\R_+)\to [0,1]$ via  $d\mu (r) = e^{-r^2} \, dr^2$, as well as the vectors 
$$
\eta^k_m (r) = (-1)^{\max\{ 0,k-m\}} \sqrt{\tfrac{(\min\{ m,k\})!}{(\max\{m,k\})!}} r^{\vert m-k\vert}  L^{\vert m-k\vert}_{\min\{ m,k\}} (r^2) \, \varphi_k \in\hil
$$
where $\varphi_k\in\hil$ is an arbitrary fixed unit vector. The normalization condition $\int \Vert \eta^k_m (r) \Vert^2 \, d\mu(r)=1$ is then satisfied, and it is merely a simple observation that $\G^{\vert k\rangle}$ takes the form of Eq.\
\eqref{eqn:structure}. Note that the angle margin $\G^{\vert k \rangle}_\textrm{angle}$ is never the canonical phase observable \cite{LaPe99,LaPe00}.
\end{example}

\begin{example}\label{ex2}\rm
A second example can be obtained by modifying the above considerations. Indeed, if we replace the vectors $\eta^k_m (r)$ by $ \xi^k_m(r) = L_{\min \{ m,k \}} (r^2) \vp_k$ while keeping the measure $\mu$ unchanged, we get an observable $\mathsf{F}^{\vert k\rangle}:\h B(\C)\to\lh$  with the explicit form 
\begin{equation}\label{eqn:ex2}
\mathsf{F}^{\vert k \rangle} (X\times \Theta) =  \sum_{m,n=0}^\infty \frac{1}{2\pi} \int_{X\times \Theta} e^{i\theta(m-n)} L_{\min \{ m,k \}} (r^2)  L_{\min \{ n,k \}} (r^2) \, e^{-r^2}\, dr^2 \, d\theta \vert m\rangle\langle n\vert. 
\end{equation}
Unlike the observable $\G^{\vert k \rangle}$, this observable is not translation covariant, though it is clearly covariant with respect to phase shifts. The significant features of these observables are in their margins. For the angle margin phase observable $\mathsf{F}^{\vert k\rangle}_\textrm{angle}$, we get the phase matrix elements
$$
c_{mn} =  \int_0^\infty L_{\min \{ m,k \}} (r^2)  L_{\min \{ n,k \}} (r^2) \, e^{-r^2}\, dr^2 = \delta_{\min \{ m,k \}, \min \{ n,k \}}.
$$
In particular, the angle margin of $\mathsf{F}^{\vert 0\rangle}$ is the canonical phase. On the contrary, if we increase the value of $k$ we find that the  observable becomes in some sense more trivial. Indeed, since $c_{mn}=\delta_{mn}$ for $m,n\leq k$, the phase matrix always contains a $(k+1)\times (k+1)$ identity matrix in the upper left corner. As for the radial margin, we get 
$$
\mathsf{F}^{\vert k\rangle}_\textrm{rad} (X) = \sum_{m=0}^k c_{mm}(X) \vert m\rangle\langle m \vert  + c_{kk}(X) \sum_{m=k+1}^\infty  \vert m\rangle\langle m \vert    
$$
which shows that by increasing the value of  $k$ the observable becomes, vaguely speaking, more and more nontrivial. We also immediately recognize the obvious fact that for $k=0$ corresponding to the angle margin being the canonical phase, the radial margin is trivial.  
\end{example}

\begin{example}\rm
As a final example, we consider the case where the probability measure is a Dirac measure $\mu= \delta_{r_0}$ concentrated at some point $r_0 >0$. The corresponding phase space observable $\mathsf P^{r_0}:\,\mathcal{B}(\C)\to\lh$ can then be written as
\begin{equation}\label{eqn:ex3}
\mathsf P^{r_0} (X \times \Theta) = \sum_{m,n=0}^\infty  \langle \eta_m\vert \eta_n \rangle \, \delta_{r_0} (X)  \frac{1}{2\pi} \int_\Theta e^{i\theta (m-n)}\, d\theta \,  \vert m\rangle\langle n\vert.
\end{equation} 
We now immediately see that the angle margin $\mathsf P^{r_0}_\textrm{angle}$ is the phase observable with the phase matrix elements $c_{mn} = \langle \eta_m \vert \eta_n\rangle$, and the radial margin is the trivial (sharp) observable $\mathsf P^{r_0}_\textrm{rad}(X)= \delta_{r_0}(X) I$.
\end{example}

\section{Double homodyne detection scheme}
We will now turn our attention to the double homodyne detector (also called the eight-port homodyne detector), a well established method for measuring translation covariant phase space observables related to a single mode electromagnetic field \cite{MQSL, KiLa2008, LaPeSc2010}. This scheme makes use of the fact that the usual single homodyne detector provides a measurement of an arbitrary field quadrature. Indeed, when the signal field under investigation is coupled to a strong auxiliary field in a coherent state $\vert z\rangle$ using a $50:50$ beam splitter, and the scaled photon number difference $\frac{1}{\vert z\vert}(I\otimes N - N\otimes I)$ of the two output modes is measured,  then for a sufficiently large $\vert z\vert$ the measured observable is approximately the quadrature observable $\sfq_\theta$, where $\theta  = \arg z$ and $\sfq_\theta (X) = e^{i\theta N} \sfq(X) e^{-i\theta N}$, $\sfq$ being the canonical spectral measure on the real line (i.e., the position observable) (for more details, see, e.g., \cite{KiLa08}).

In double homodyne detection, the signal field is first coupled to a parameter field in some state $\sigma$ via a $50:50$ beam splitter, after which a phase shift of $-\frac{\pi}{2}$ is performed on one of the output modes (see Fig.\ \ref{eightport}). Balanced homodyne detection is then performed on each output mode, so that by choosing the phase of the auxiliary coherent field to be zero in both measurements, this corresponds to measuring the canonical spectral measure on $\R^2$ for the two-mode field. With this setup, the measured observable is the phase space observable $\G^{\sigma'}$ where the generating operator is connected to the state of the parameter field via the conjugation map $(C\varphi)(x) = \overline{\varphi(x)}$ as $\sigma'=C\sigma C$  \cite{KiLa2008}. In particular, if we want the observable to be phase shift covariant and thus give rise to a phase observable as the angle margin, we must use a parameter field which is diagonal in the number state representation, $\sigma = \sum_{k=0} \lambda_k \vert k\rangle\langle k \vert$. The simplest case is obviously obtained by using the vacuum $\sigma= \vert 0 \rangle\langle 0 \vert$.

\begin{figure}[h!]
\includegraphics[width=0.7\columnwidth]{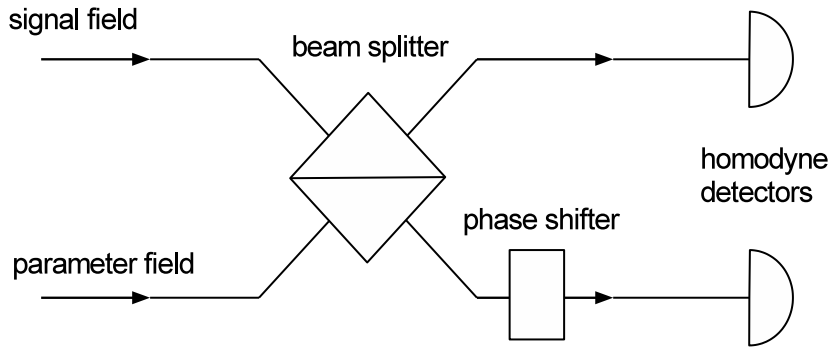}
\caption{Schematic of a double homodyne detector. The signal field is mixed with a parameter field by means of a $50:50$ beam splitter, after which a phase shift of $-\pi/2$ is performed on one of the modes. Balanced homodyne detection is then performed on both modes.}\label{eightport}
\end{figure}

The problem with using this measurement setup is of course caused by the fact that the canonical phase cannot be obtained as the margin of any translation covariant phase space observable \cite{LaPe99,LaPe00}. Thus, a modification of the setup is needed. We will next show that a suitable modification is obtained by adding a unitary coupling between the signal and parameter fields prior to the beam splitter. In other words, while in the usual double homodyne detection the signal and parameter fields are uncorrelated before entering the beam splitter, with this modification they will enter the beam splitter in an entangled state. With this method, it is in principle possible to measure also the observables $\mathsf{F}^{\vert k \rangle}$ encountered in Example \ref{ex2}.

To this end, first notice that in the usual double homodyne detector, the overall unitary coupling consisting of the beam splitter and the phase shifter is given by $U:\,\hi_{\rm in}\to\hi_{\rm out}$, 
\begin{equation}\label{eqn:8port_coupling}
\big(U  \vert m\rangle \otimes \vert n \rangle \big) (r,\theta)   =  \frac{1}{\sqrt{\pi}} \, \langle n \vert D(re^{i\theta})^* \vert m \rangle,
\end{equation}
where $\hi_{\rm in}\simeq\hi\otimes\hi$ consists of the input signal and parametric field modes and $\hi_{\rm out}\simeq\hi\otimes\hi$   is the output space. The observable measured with the two homodyne detectors is  then (or rather may be chosen to be) the canonical spectral measure $\sfq:\h B(\C)\to\h L(\hil_{\rm out})$. If $\rho$ and $\sigma$ are the states of the signal and parameter fields, respectively, then a direct computation shows that 
$$
\tr{U(\rho\otimes \sigma )U^* \sfq (Z) } = \tr{\rho \G^{\sigma'} (Z)}
$$
for all $Z\in\h B(\C)$. Hence, by considering only the angle margin $\sfq_{\rm angle}$ as the pointer observable, one can measure the angle margin $\G^{\sigma'}_{\rm}$. In particular, by preparing the parameter field in a state which is diagonal in the number basis, one  can measure any phase space phase observable.

Now it is easily seen from the discussion in Example \ref{ex2} that the total coupling needed for the measurement of $\mathsf{F}^{\vert k \rangle}$ is 
\begin{equation}\label{eqn:total_coupling}
\big( V \vert m\rangle\otimes \vert n\rangle \big) (r,\theta) = \frac{1}{\sqrt{\pi}} e^{i\theta(n-m)} L_{\min \{m,n\}} (r^2) e^{-\frac{r^2}{2}}.
\end{equation}
We want to exploit the fact that we already have at our disposal the unitary coupling \eqref{eqn:8port_coupling} of the usual double homodyne detector. Therefore, we will look for a unitary operator $W:\hil_{\rm in} \to\hil_{\rm in}$ such that $V = UW$, which would then amount to adding an extra component to the measurement setup prior to the beam splitter, see Fig.\ \ref{eightport_mod1}. \footnote{One could of course consider equally well a unitary operator $W':\hil_{\rm out} \to\hil_{\rm out}$ such that $V = W'U$, but this leads to a similar treatment.}

\begin{figure}[h!]
\includegraphics[width=0.8\columnwidth]{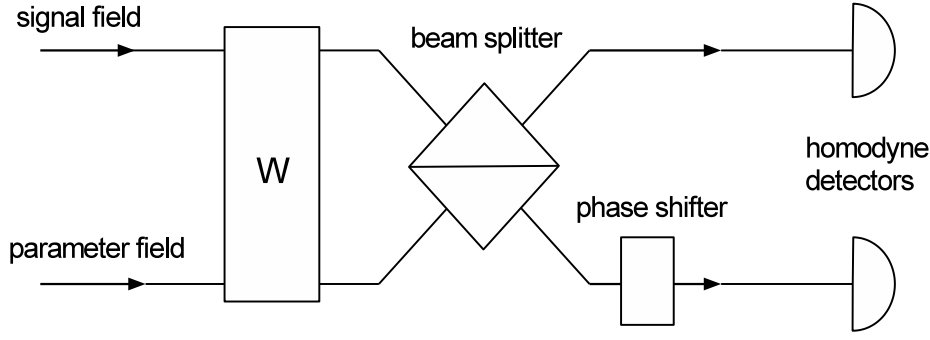}
\caption{A modified double homodyne detector where an additional unitary coupling $W$ is performed prior to the beam splitter. By choosing the parameter field to be in the vacuum state, on can obtain the canonical phase observable as the angle margin of the measured phase space observable.}\label{eightport_mod1}
\end{figure}

The action of the operator  $W$ is now given by
$$
W(\vert m\rangle\otimes \vert n \rangle)  = \sum_{k,l=0}^\infty \Big\langle U \vert k \rangle \otimes \vert l \rangle \Big|  V \vert m \rangle \otimes \vert n\rangle \Big\rangle \, \vert k \rangle \otimes \vert  l\rangle .
$$
Let us denote $\alpha_{kl,mn}= \big\langle U\vert k \rangle \otimes \vert l \rangle \big| V \vert m \rangle \otimes \vert n\rangle \big\rangle$ so that
{\small
\begin{align}
\alpha_{kl,mn}  &= (-1)^{\max\{ 0,l-k\}} \sqrt{\tfrac{\min\{ k,l\}!}{\max\{k,l\}!}}  \, \frac{1}{\pi} \int_0^\infty \int_0^{2\pi} e^{i\theta(k-l+n-m)}r^{\vert k-l\vert}  L^{\vert k-l\vert}_{\min\{ k,l\}} (r^2) L_{\min \{ m,n \}}(r^2)e^{- r^2} \, rdrd\theta\nonumber\\
= &\; \delta_{l,k+n-m} (-1)^{\max\{ 0,n-m\}} \sqrt{\tfrac{\min\{ k,k+n-m\}!}{\max\{k,k+n-m\}!}}  \,  \int_0^\infty x^{\frac{1}{2}\vert n-m\vert}  L^{\vert n-m\vert}_{\min\{ k,k+n-m\}} (x) L_{\min \{ m,n \}}(x)e^{- x} \, dx.\nonumber
\end{align}
}
The above integral can be further calculated, but it seems that a simple closed expression does not exist. However, note that this already tells us that $W(\vert m\rangle\otimes \vert n\rangle )$ is an eigenvector of the photon number difference operator $\Delta N$. Indeed, $\Delta N\,W(\vert m\rangle\otimes \vert n\rangle )=(m-n)W(\vert m\rangle\otimes \vert n\rangle )$ and
$
W(\vert m\rangle\otimes \vert n\rangle )
$
can be written as a series with respect to the RNSs $|m-n,s\>\>$, $s\in\N$.

Since the canonical phase is obtained as the special case $n=0$ of the above consideration, we will next focus on that. Now the above integrals are easily calculated and we obtain the expressions $W(\vert 0\rangle\otimes \vert 0 \rangle)=\vert 0\rangle\otimes \vert 0\rangle$ and, for $m>0$,
\begin{eqnarray}\nonumber
W(\vert m\rangle\otimes \vert 0 \rangle)  &=& \sum_{k=0}^\infty \frac{\Gamma(k+\tfrac{m}{2}+1)}{\sqrt{k!(k+m)!}}\frac{m}{2k+m} \vert k +m \rangle \otimes \vert k\rangle \\
&=&  \frac{m}{2}\sum_{k=0}^\infty \frac{\Gamma(k+\tfrac{m}{2})}{\sqrt{k!(k+m)!}} |m,k\>\>
=
\frac{m}{2{\sqrt{m!}}}\sum_{k=0}^\infty {\Gamma(k+\tfrac{m}{2})}{m+k\choose m}^{-1/2} |m,k\>\> \label{eqn:W_expression}
\end{eqnarray}
where $\vert k +m \rangle \otimes \vert k\rangle$ is the RNS $|m,k\>\>$.
From $V=UW$ one sees that the canonical phase distribution $\Theta\mapsto\tr{\rho\Ecan(\Theta)}$ of any state $\rho$ can be measured directly
by first producing an entangled two-mode input state $S(\rho)=W(\rho\otimes|0\>\<0|)W^*$ and then  measuring $S(\rho)$ via double  homodyne detection.

\begin{remark}\rm
Note that $W(\vert m\rangle\otimes \vert 0 \rangle)$ is a two mode nonlinear coherent state \cite{Wang} which are generally of the form
$
|\alpha,f,m\>=\sum_{k=0}^\infty C_k\vert k +m \rangle \otimes \vert k\rangle 
$
where the constants $C_k\in\C$ are such that $\sum_k|C_k|^2=1$.
Moreover, they satisfy the equations
$$
f(N\otimes I,I\otimes N)(a\otimes a)|\alpha,f,m\>=\alpha|\alpha,f,m\>,\qquad \alpha\in\C,
$$
where $f(N\otimes I,I\otimes N)$ is some fixed function of the single-mode number operators. As shown in \cite{Wang}, 
$$
C_k=\alpha^k\sqrt{\frac{m!}{k!(k+m)!}}\left[
\prod_{s=1}^k\frac{1}{f(s-1+m,s-1)}
\right]C_0
$$
so that $W(\vert m\rangle\otimes \vert 0 \rangle)=|\alpha,f,m\>$ where 
$\alpha=1$, $C_0={\Gamma(\tfrac{m}{2}+1)}/{\sqrt{m!}}$, and
$
f(n_1,n_2)=2/(n_1+n_2),
$
$n_1,\,n_2\in\N$, $n_1\ne 0$. 
\end{remark}

\begin{remark}\rm\label{remu2}
Note that there does not exist bounded (or unitary) operators $\widetilde D(re^{i\theta})$ such that
$\langle n \vert \widetilde D(re^{i\theta})^* \vert m \rangle=e^{i\theta(n-m)} L_{\min \{m,n\}} (r^2) e^{-\frac{r^2}{2}}$ for all $r$ and $\theta$. However, we can write
$$
\big( V\vert m\rangle\otimes \vert n\rangle \big) (r,\theta) =
\frac{1}{\sqrt{\pi}}\big\<\big\<\theta,r^2\big|\,|m\>\otimes|n\>\,\big\>
$$
where $|\theta,x\>\>:=\sum_{n=-\infty}^\infty \sum_{k=0}^\infty e^{in\theta} L_k(x)e^{-x/2}|n,k\>\>$, $\theta\in[0,2\pi)$, $x\ge 0$, is (essentially) a generalized vector introduced by Ban \cite[Eq.\ (19b)]{Ban2}, see Eq.\ \eqref{Banban}. Now the angle margin observable of the (two mode) spectral measure
$$
\mathsf O:\,\mathcal B(\C)\to\mathcal L(\hi\otimes\hi),\qquad
Z\mapsto\mathsf O(Z):= V^*\sfq(Z) V=\frac1\pi\int_Z|\theta,r^2\>\>\<\<\theta,r^2| rdrd\theta
$$
is $\mathsf O_{\rm angle}$ defined in Eq.\ \eqref{Qpilkku}, i.e.\ $\mathsf O(\R_+\times\Theta)=\mathsf O_{\rm angle}(\Theta)$.
The radial margin is
$$
\mathsf O_{\rm rad}(X):=
\mathsf O(X\times[0,2\pi))=\sum_{n=-\infty}^\infty\sum_{k,l=0}^\infty\int_X L_k(r^2)L_l(r^2)e^{-r^2}d r^2|n,k\>\>\<\<n,l|.
$$
Immediately one sees that $\tr{(\rho\otimes|k\>\<k|)\mathsf O(Z)}=\tr{\rho\F^{|k\>}(Z)}$.
\end{remark}

\begin{example}\rm
The action of $V$ is easy to calculate for some physically relevant states. For example, if the first (or the second) mode is in the coherent state $|\alpha\>=e^{-|\alpha|^2/2}\sum_{m=0}^\infty\alpha^m/\sqrt{m!}|m\>$
and the second mode is in the vacuum state $|0\>$ then
$$
(V|\alpha\>\otimes|0\>)(r,\theta)=\frac{e^{-|\alpha|^2/2}}{{\sqrt{\pi}}}\sum_{m=0}^\infty\frac{\alpha^m}{\sqrt{m!}}e^{-i\theta m} e^{-\frac{r^2}{2}}
=\frac{e^{-r^2/2}}{{\sqrt{\pi}}}\<\theta|\alpha\>,
$$
that is, we obtain the London distribution $\theta\mapsto\<\theta|\alpha\>$ of the coherent state.
Let then
$$
\psi_\alpha:=C(\alpha)\sum_{m=0}^\infty\frac{\alpha^m}{m!}|m\>\otimes|m\>,\qquad\alpha\in\C
$$
be a pair-coherent state where $C(\alpha):=J_0(2i|\alpha|)^{-1/2}$ is a normalization constant and 
$$
J_0(x)=\sum_{m=0}^\infty\frac{(-1)^m}{(m!)^2}\left(\frac{x}{2}\right)^{2m}
$$ 
is the zeroth Bessel function of the first kind. Now
$$
(V\psi_\alpha)(r,\theta)=\frac{C(\alpha)}{\sqrt{\pi}}\sum_{m=0}^\infty\frac{\alpha^m}{m!}L_m(r^2)e^{-r^2/2}
=\frac{C(\alpha)}{\sqrt{\pi}}J_0\big(2r\sqrt{\alpha}\big)e^{\alpha}e^{-r^2/2}
$$
whereas, for a `two-mode phase coherent state'
$\widetilde\psi^q_\alpha:=(1-|\alpha|^2)^{1/2}\sum_{m=0}^\infty{\alpha^m}|m\>\otimes|m+q\>,$ $q\in\N$, $\alpha\in\C$, $|\alpha|<1$, one gets the Gaussian function
$$
\big(V\widetilde\psi^q_\alpha\big)(r,\theta)=\frac{(1-|\alpha|^2)^{1/2}}{\sqrt{\pi}}e^{iq\theta}\sum_{m=0}^\infty\alpha^m L_m(r^2)e^{-r^2/2}
=
\frac{(1-|\alpha|^2)^{1/2}}{\sqrt{\pi}(1-\alpha)}e^{iq\theta-[\alpha/(1-\alpha)+1/2]r^2}
$$
(see Eqs.\ 8.975(3) and 8.975(1) of \cite{ruskeakirja}).
Finally, we note that, by using $V$, it is possible to produce the (weighted) monomial distributions $r^{2n}e^{-r^2/2}$ in the radial direction: Indeed, the (nonnormalized) state $\vp_n:=\sum_{m=0}^n(-1)^m{n\choose m}|m\>\otimes|m\>$ gives
$$
(V\vp_n)(r,\theta)=\frac{1}{\sqrt{\pi}}\sum_{m=0}^n(-1)^m{n\choose m}L_m(r^2)e^{-r^2/2}=\frac{1}{\sqrt{\pi}n!}r^{2n}e^{-r^2/2}.
$$
\end{example}

\section{Conclusions}
We have studied measurements of covariant phase observables, with the aim of obtaining a realistic measurement model for the canonical phase. To this end, we have first determined the minimal measurement models of phase observables. However, since practical issues suggest that a realistic phase measurement is most easily obtained by measuring a phase space observable, we have considered such joint measurements. As a special case, we have focused on the so-called phase shift covariant phase space observables. In particular, we have shown that the canonical phase may be obtained as the angle margin of such an observable. Finally, we have considered the quantum optical double homodyne detection scheme, and its modification  as  a means of measuring phase space observables. We have constructed a unitary coupling which, when placed in front of the setup, would allow one to measure the canonical phase observable.

\vspace{1cm}

\noindent{\bf Acknowledgments.} 
The authors thank Pekka Lahti for comments and suggestions. Financial support from the Academy of Finland (grant no. 138135.) and the Italian Ministry of Education, University and Research (FIRB project RBFR10COAQ) is gratefully acknowledged.

\end{document}